\documentclass[preprint,12pt]{elsarticle}

\usepackage{subfigure} 
\usepackage{psfig}
\usepackage{amsmath}
\usepackage{graphicx}
\usepackage{amssymb}
\usepackage{multirow}
\usepackage{xcolor}

\usepackage{amssymb}
\usepackage{amsthm}

\newtheorem{prop}{Proposition}[section]

\newtheorem{remark}{Remark}[section]

\usepackage{lineno}

\begin{document}

\begin{frontmatter}



\title{FORM-based global reliability sensitivity analysis of systems with multiple failure modes}

 \author{Iason Papaioannou}
 \author{Daniel Straub}
 \address{Engineering Risk Analysis Group, Technische Universit{\"a}t M{\"u}nchen, Arcisstr. 21, 80290 M{\"u}nchen, Germany}
 


\begin{abstract}
Global variance-based reliability sensitivity indices arise from a variance decomposition of the indicator function describing the failure event. The first-order indices reflect the main effect of each variable on the variance of the failure event and can be used for variable prioritization; the total-effect indices represent the total effect of each variable, including its interaction with other variables, and can be used for variable fixing. This contribution derives expressions for the variance-based reliability indices of systems with multiple failure modes that are based on the first-order reliability method (FORM). The derived expressions are a function of the FORM results and, hence, do not require additional expensive model evaluations. 
They do involve the evaluation of multinormal integrals, for which effective solutions are available. 
We demonstrate that the derived expressions enable an accurate estimation of variance-based reliability sensitivities for general system problems to which FORM is applicable.

\end{abstract}

\begin{keyword}
Reliability analysis \sep Sensitivity analysis \sep system reliability \sep FORM


\end{keyword}

\end{frontmatter}


\section{Introduction}
\label{S:1}

Structural reliability analysis evaluates the probability of failure of a structural component or system accounting for uncertainties in the structural parameters and external loading \cite{der2022structural}.
Failure of structural systems can be described by a combination of multiple failure modes.
The estimation of the system failure probability is a challenging task as it involves solving a high-dimensional probability integral over a failure domain that is often disjoint and has non-smooth boundaries. 
The failure domain is only known point-wise through the evaluation of one or several engineering models, which can be computationally intensive.
Application of standard Monte Carlo is inefficient when the target probability is small, as is typically the case for failure probabilities.
Rare event simulation methods employ tailored sampling techniques to efficiently sample the failure domain and estimate the sought probability; they include subset simulation \cite{Au01}, sequential importance sampling \cite{Papaioannou16,cheng2023rare}, and the cross-entropy method \cite{Kurtz13,papaioannou2019improved}.
Conversely, approximation methods seek a computable approximation of the failure probability through performing a Taylor approximation of the boundary of the failure domain at one or more design points, and include the first- and second-order reliability methods (FORM/SORM) \cite{hohenbichler1982first,hohenbichler1987new,ADK05}.

In the design and assessment of engineering systems, it is often essential to understand the influence of the uncertain parameters on the probability of failure, a task known as reliability sensitivity analysis.
Local reliability sensitivity analysis involves evaluation of the partial derivatives of the probability of failure with respect to uncertain or deterministic input parameters \cite[e.g.,][]{Wu94,Song09,Papaioannou13,dubourg2014meta,jensen2015reliability,papaioannou2018reliability}, Global reliability sensitivity analysis examines the average effects of the uncertain inputs on the probability of failure \cite[e.g.,][]{au2004probabilistic,cui2010moment,luyi2012moment,wei2012efficient,wang2013application,borgonovo2016moment,ehre2020framework,idrissi2021developments,straub2022decision}. 
Particularly, variance-based reliability sensitivity analysis, introduced in \cite{wei2012efficient}, employs the variance decomposition of the indicator function of the failure event to determine the proportion of the variance that can be attributed to each input and subsets thereof.
The resulting sensitivity measures are extensions of classical variance-based sensitivity analysis \cite{sobol1993sensitivity} and include the first-order and total-effect reliability sensitivity indices.
The first-order indices indicate the contribution of the variance of individual inputs to the variance of the indicator function and can be used for variable prioritization, whereas total-effect indices represent the total contribution of all variance terms that include a certain input and can be used for variable fixing. 
Several methods have been proposed to estimate the first-order and total effect indices, including single-loop sampling methods \cite{wei2012efficient}, the state dependent parameter method \cite{luyi2012moment} and an approach that post-processes failure samples from sampling-based reliability methods \cite{perrin2019efficient,li2019global}.
Extensions of these indices for problems with dependent inputs were recently proposed in \cite{ehre2024variance}.
A more traditional sensitivity measure that is widely used in structural reliability is the ${\alpha}$-factors, obtained as a by-product of a FORM analysis  \cite{ADK05,hohenbichler1986sensitivity}.
The ${\alpha}$-factors are the directional cosines of the design point, i.e., the point in the failure domain with maximum likelihood.
Generalizations of the ${\alpha}$-factors have been proposed for dependent inputs in \cite{ADK05} and for multimodal failure domains in \cite{kim2018generalized}.

In \cite{papaioannou2021variance}, the authors derived approximations of the variance-based sensitivity indices that depend on results of a FORM analysis and studied the relation of these approximations to the $\alpha$-factors. 
These approximations work well for structural reliability problems with a single dominant failure mode. 
In this contribution, we derive approximations of the reliability sensitivity indices that are based on the FORM approximation to system problems with multiple failure modes. 
We first study the FORM approximations to parallel and series systems and derive expressions for the reliability sensitivities of these approximations.
Thereafter, we apply the inclusion-exclusion principle to deal with general systems.
We demonstrate with numerical examples that the derived formulas provide excellent approximations of the variance-based reliability sensitivity indices for general system problems.

The structure of the paper is as follows. In Section \ref{s:problem}, we review system reliability analysis and the variance-based reliability sensitivity indices. Section \ref{s:form} discusses FORM for system reliability analysis. Section \ref{s:sensitivities} presents approximations of the variance-based sensitivities for series and parallel problems as well as for general systems. Section \ref{s:examples} presents numerical examples that test the accuracy of the FORM approximations. The paper closes with the conclusions in Section \ref{s:conclusion}.

\section{Problem formulation} \label{s:problem}

\subsection{System reliability}

Consider an engineering system with $m$ different component failure modes.
Let $\boldsymbol{X}$ denote a continuous random vector of dimension $n$ modeling the uncertain system variables, which is described by a joint probability density function (PDF) $f(\boldsymbol{x})$.
Each failure mode is mathematically represented by a limit-state function (LSF) $g_i(\boldsymbol{x})$ such that the corresponding component failure event is $F_i=\{g_i(\boldsymbol{X}) \leq 0\}$ with $i=1,\ldots,m$.
Failure of the system may occur due to different combinations of occurrence of the individual failure events $F_i$. 
The series and parallel system problems are defined, respectively, by the following events:
\begin{equation} \label{eq:Fser}
F_\mathrm{ser} = \bigcup_{i=1}^m F_i 
\end{equation}
and 
\begin{equation} \label{eq:Fpar}
F_\mathrm{par} = \bigcap_{i=1}^m F_i  \, .
\end{equation}
The general system problem can be defined as a series of parallel system problems, as follows
\begin{equation} \label{eq:Fgensys}
F_\mathrm{sys} = \bigcup_{k=1}^K \bigcap_{i \in \boldsymbol{c}_k} F_i \, ,
\end{equation}
where $\boldsymbol{c}_k \in \mathcal{P}(\{1,\ldots,m\}) \backslash \{ \emptyset \}$ are the cut-sets of the system and $\mathcal{P}(S)$ denotes the power set of $S$.
From Eq.~\eqref{eq:Fgensys} one can retrieve the component problem by setting $K = 1$ and $\boldsymbol{c}_1= \{i\}$, the parallel system problem by setting $K = 1$ and $\boldsymbol{c}_1= \{1,\ldots,m\}$, and the series system problem by setting $K = m$ and $\boldsymbol{c}_k= \{k\}$.

The probability of failure of the system can be expressed as
\begin{equation} \label{eq:pf}
p_{F,\mathrm{sys}}=\Pr(F_\mathrm{sys}) = \int_{\Omega_{F,\mathrm{sys}}}{ f (\boldsymbol{x})\mathrm{d}\boldsymbol{x}} = \int_{\mathbb{R}^n}{ \mathrm{I}_{F,\mathrm{sys}}(\boldsymbol{x}) f (\boldsymbol{x})\mathrm{d}\boldsymbol{x}} \, ,
\end{equation}
where $\Omega_{F,\mathrm{sys}}=\{\boldsymbol{x} \in \mathbb{R}^n : \bigcup_{k=1}^K \bigcap_{i \in \boldsymbol{c}_k} ( g_i(\boldsymbol{x}) \leq 0 ) \}$ is the system failure domain and $\mathrm{I}_{F,\mathrm{sys}}(\boldsymbol{x})$ is the indicator function of the system failure event, which takes the value 1 if $\boldsymbol{x} \in \Omega_{F,\mathrm{sys}}$ and 0 otherwise.
In general system problems the failure domain often has complex shape and is only known point-wise in terms of one or more computationally demanding engineering models.
Moreover, the target probability $p_{F,\mathrm{sys}}$ of safety-critical engineering systems is often small.
Hence, accurate estimation of $p_{F,\mathrm{sys}}$ is a nontrivial task as it requires the efficient estimation of rare events in complex failure domains.

\subsection{Variance-based reliability sensitivity analysis} \label{sec:VBrelSens}


Variance-based sensitivity analysis aims at identifying the input random variables in $\boldsymbol{X}$ that have largest impact on the variance of a quantity of interest (QOI) $Q = h(\boldsymbol{X})$, where $h:\mathbb{R}^n \to \mathbb{R}$ defines an input-output relationship.
It is based on the functional analysis of variance (ANOVA) decomposition of $h(\boldsymbol{x})$.
Consider the case where the random vector $\boldsymbol{X}$ consists of statistically independent components, i.e., $f(\boldsymbol{x})=\prod_{i=1}^n f_i(x_i)$ with $f_i(x_i)$ denoting the marginal PDF of $X_i$.
Also assume that $h(\boldsymbol{x})$ is square-integrable, i.e., $ \mathrm{E} [h(\boldsymbol{X})^2] < \infty$.
The functional ANOVA decomposition of $h(\boldsymbol{x})$ reads \cite{efron1981jackknife,sobol1993sensitivity}:
\begin{equation} \label{eq:FANOVA}
h(\boldsymbol{x}) = h_{\emptyset} + \sum_{i=1}^n h_i (x_i) +\sum_{1 \leq i < j \leq n} h_{i,j} (x_i,x_j) + \cdots + h_{1,\ldots, n} (x_1,\ldots,x_n) \, .
\end{equation}
The representation of Eq.~\eqref{eq:FANOVA} exists and is unique provided that 
\begin{equation} \label{eq:FANOVAcond}
\mathrm{E}[ h_{\boldsymbol{v}}(\boldsymbol{X}_{\boldsymbol{v}} ) \vert \boldsymbol{X}_{\boldsymbol{v} \backslash i} ]= \int_{-\infty}^{\infty} h_{\boldsymbol{v}}(\boldsymbol{x}_{\boldsymbol{v}} )  f_i(x_i) \mathrm{d} x_i = 0, \; \forall i \in \boldsymbol{v}, \forall \boldsymbol{v} \in \mathcal{P}(\{1,\ldots,n\}) \, ,
\end{equation}
with $\boldsymbol{x}_{\boldsymbol{v}} = \{x_i, i \in \boldsymbol{v}\}$. From Eq.~\eqref{eq:FANOVAcond} it follows that $h_{\emptyset} = \mathrm{E}[ h(\boldsymbol{X} ) ]$ and that the summands in Eq.~\eqref{eq:FANOVA} are mutually orthogonal, i.e., it is $\mathrm{E}[ h_{\boldsymbol{v}}(\boldsymbol{X}_{\boldsymbol{v}} ) h_{\boldsymbol{w}}(\boldsymbol{X}_{\boldsymbol{w}} ) ] =0 $ for $\boldsymbol{v} \neq \boldsymbol{w} \in \mathcal{P}(\{1,\ldots,n\})$.
From the orthogonality property, one gets the following decomposition of the variance of $Q$ in terms of the variances of the ANOVA summands:
\begin{equation} \label{eq:ANOVA}
\mathrm{Var}(Q) = \sum_{i=1}^n V_i +\sum_{1 \leq i < j \leq n} V_{i,j} + \cdots + V_{1,\ldots, n} \, ,
\end{equation}
where $V_{\boldsymbol{v}}=\mathrm{Var}(h_{\boldsymbol{v}}(\boldsymbol{X}_{\boldsymbol{v}} ) )$. 
The component variances can also be expressed recursively in terms of variances of conditional expectations as
\begin{equation} \label{eq:ViCondExp}
V_{\boldsymbol{v}} = \mathrm{Var}(\mathrm{E}[Q \vert \boldsymbol{X}_{\boldsymbol{v}} ]) - \sum_{\boldsymbol{w} \in \mathcal{P}(\boldsymbol{v}) \backslash \{\emptyset,\boldsymbol{v} \} } V_{\boldsymbol{w}} \, .
\end{equation}
Assume now that $\mathrm{Var}(Q) \neq 0$. The Sobol' index $S_{\boldsymbol{v}}$ \cite{sobol1993sensitivity} and total-effect index $S_{\boldsymbol{v}}^T$ \cite{homma1996importance} associated with variable subset $X_{\boldsymbol{v}}$ are defined as
\begin{equation} \label{eq:Sobol}
S_{\boldsymbol{v}} = \frac{V_{\boldsymbol{v}}}{\mathrm{Var}(Q)} \, , \qquad S_{\boldsymbol{v}}^T = 1- \frac{\mathrm{Var}(\mathrm{E}[Q \vert \boldsymbol{X}_{\sim \boldsymbol{v}} ]) }{\mathrm{Var}(Q)} \, ,
\end{equation}
where $\sim \boldsymbol{v} = \{1,\ldots,n \} \backslash \boldsymbol{v}$.
The Sobol' index $S_{\boldsymbol{v}}$ measures the portion of $\mathrm{Var}(Q)$ due to the interactions between variables $\boldsymbol{X}_{\boldsymbol{v}}$, whereas the total-effect index $S_{\boldsymbol{v}}^T$ measures the contribution due to variables $\boldsymbol{X}_{\boldsymbol{v}}$ and their interactions with all other variables in $\boldsymbol{X}$.
For $\boldsymbol{v}=\{ i \}$, the first-order Sobol' index $S_i$ measures the contribution of the main effect of $X_i$ on the variance of $Q$ and the index $S_{i}^T$ reflects the variance portion due to $X_i$ including its higher-order interactions, i.e., it is $S_{i}^T = \sum_{\boldsymbol{v} \in \mathcal{P}(1,\ldots,n), i \in  \boldsymbol{v} } S_{\boldsymbol{v}}$. 

\begin{remark} \label{rem:closed}
Alternative to the Sobol' index $S_{\boldsymbol{v}}$ of Eq.\eqref{eq:Sobol}, one can evaluate the closed Sobol' index $S_{\boldsymbol{v}}^{clo}$ of variable subset $X_{\boldsymbol{v}}$.
$S_{\boldsymbol{v}}^{clo}$ collects the total variance contribution of all variables in $X_{\boldsymbol{v}}$ and can be defined as $S_{\boldsymbol{v}}^{clo}=\mathrm{Var}(\mathrm{E}[Q \vert \boldsymbol{X}_{\boldsymbol{v}} ])$.
We note that the closed Sobol' index of a single variable $X_i$ is identical to the corresponding Sobol' index.
Closed Sobol' indices are particularly relevant if a set of variables represent the effect of a single physical quantity, e.g., the variables entering the discrete representation of a random field representing the spatial variability of a material property \cite[e.g.,][]{ehre2020global}.
\end{remark}

In \cite{wei2012efficient}, variance-based sensitivity indices are proposed that are based on the variance decomposition of the random variable $Z = \mathrm{I}_{F}(\boldsymbol{X})$, describing the geometry of the failure domain in the outcome space of $\boldsymbol{X}$.
To simplify notation, we omit the subscript that denotes the type of failure event; the reliability sensitivity indices discussed herein apply to any component or system failure event.
The variable $Z$ follows the Bernoulli distribution with parameter $p_F=\Pr(F)$; it has mean $\mathrm{E}[Z]=p_F$ and variance $\mathrm{Var}(Z)=p_F(1-p_F)$.
A decomposition of the variance of $Z$ leads to the following Sobol' reliability sensitivity index:
\begin{equation} \label{eq:SobolRel}
S_{F,\boldsymbol{v}}=\frac{V_{F,\boldsymbol{v}}}{\mathrm{Var}(Z)}=\frac{V_{F,\boldsymbol{v}}}{p_F(1-p_F)} \, ,
\end{equation}
with
\begin{equation} \label{eq:ViCondExpF}
V_{F,\boldsymbol{v}} = \mathrm{Var}(\mathrm{E}[Z \vert \boldsymbol{X}_{\boldsymbol{v}} ]) - \sum_{\boldsymbol{w} \in \mathcal{P}(\boldsymbol{v}) \backslash \{\emptyset,\boldsymbol{v} \} } V_{F,\boldsymbol{w}} \, .
\end{equation}
The first term in the right hand side of Eq.~\eqref{eq:ViCondExpF} can also be written as the variance of the conditional probability of $F$ given $\boldsymbol{X}_{\boldsymbol{v}}$, i.e.,  $\mathrm{Var}(\mathrm{E}[Z \vert \boldsymbol{X}_{\boldsymbol{v}} ]) =\mathrm{Var}(\Pr(F \vert \boldsymbol{X}_{\boldsymbol{v}} ))$ \cite{luyi2012moment}.
The total-effect reliability sensitivity index is given by
\begin{equation} \label{eq:totalEffectRel}
S_{F,\boldsymbol{v}}^T = 1- \frac{\mathrm{Var}(\mathrm{E}[Z \vert \boldsymbol{X}_{\sim \boldsymbol{v}} ]) }{\mathrm{Var}(Z)} = 1- \frac{\mathrm{Var}(\Pr(F \vert \boldsymbol{X}_{\sim \boldsymbol{v}} )) }{p_F(1-p_F)} \, .
\end{equation}
The first-order reliability component index $S_{F,i}$ can be used for factor prioritization, i.e., to identify which random variable $X_i$ if learned (e.g., through investing in measurement campaigns) will have the largest impact on the value of $p_F$.
The total-effect reliability component index $S_{F,i}^T$ can be used for variable fixing, i.e., to identify the random variables with $S_{F,i}^T \approx 0$, which if fixed will not impact the prediction of $p_F$.

The Sobol' and total-effect reliability sensitivity indices of Eqs. \eqref{eq:SobolRel} and \eqref{eq:totalEffectRel} can be estimated by several sampling-based approaches \cite[e.g.,][]{luyi2012moment,wei2012efficient,perrin2019efficient,li2019global}.
As an alternative, in \cite{papaioannou2021variance} we derive expressions for the sensitivity indices of the FORM probability approximation of component problems with a unique design point. 
In the following, we derive approximations of these indices based on the FORM approach to system reliability analysis. 

\begin{remark} \label{rem:dependent}
Recently, the first-order and total-effect reliability sensitivity indices were extended to treat reliability problems with dependent inputs \cite{ehre2024variance}.
This is achieved by transforming the input space to an equivalent independent space using the Rosenblatt transform and evaluating the indices of the independent variables. 
Through performing $n$ cyclic shifts of the input vector $\boldsymbol{X}$ and corresponding transformations it is possible to isolate the first-order and total-effect contribution of each variable excluding its dependence with other variables.
Although in this paper we focus on problems with independent inputs, the concepts discussed herein can be directly applied to problems with dependent inputs through application of the approach in \cite{ehre2024variance}.
\end{remark}

\section{FORM for system problems} \label{s:form}

FORM is an approximation method for solving structural reliability problems that is based on a linearization of the LSF(s) describing the failure event.
Although originally developed for component reliability problems, it has been extended for application to reliability problems with multiple modes of failure in \cite{hohenbichler1982first}.
Therein, a FORM approximation is developed for the probabilities of the elementary series and parallel system failure events of Eqs.~\eqref{eq:Fser} and \eqref{eq:Fpar}.

The first step of FORM is to transform the problem to an equivalent random variable space $\boldsymbol{U}$, consisting of independent standard normal random variables.
The vector $\boldsymbol{U}$ can be expressed in terms of the original random vector $\boldsymbol{X}$ through an isoprobabilistic mapping $\boldsymbol{U} = \boldsymbol{T} (\boldsymbol{X})$ \cite{Hohe81,ADK86}.
For the case where the variables $\boldsymbol{X}$ are statistically independent and have strictly increasing marginal cumulative distribution functions (CDFs) $F_i(x_i),i=1,\ldots,n$, this mapping is $ \boldsymbol{T}(\boldsymbol{x}) = \Phi^{-1}[F_1(x_1);\cdots;F_n(x_n)]$, with $\Phi$ denoting the standard normal CDF.
Each LSF $g_i(\boldsymbol{x})$, describing the component failure event $F_i$, can be expressed in the $\boldsymbol{U}$-space as $G_i(\boldsymbol{u}) = g_i[\boldsymbol{T}^{-1}(\boldsymbol{u})]$ with $\boldsymbol{T}^{-1}$ denoting the inverse mapping $\boldsymbol{X} = \boldsymbol{T}^{-1} (\boldsymbol{U})$.
Each component failure event is expressed in the $\boldsymbol{U}$-space as $F_i = \{ G_i(\boldsymbol{U}) \leq 0 \}$ and the probability integral of Eq.~\eqref{eq:pf}, describing the general system problem, can be transformed as:
\begin{equation} \label{eq:pfsns}
p_{F,\mathrm{sys}}=\Pr(F_\mathrm{sys}) = \int_{\Omega_{F,\mathrm{sys}}}{ \varphi_n (\boldsymbol{u}) \mathrm{d}\boldsymbol{u}}= \int_{\mathbb{R}^n} \mathrm{I}_{F,\mathrm{sys}}(\boldsymbol{u}) {\varphi_n (\boldsymbol{u})\mathrm{d}\boldsymbol{u}} \, ,
\end{equation}
where $\varphi_n$ is the $n$-variate independent standard normal PDF and $\mathrm{I}_{F,\mathrm{sys}}(\boldsymbol{u})$ is a function that indicates membership of the outcome $\boldsymbol{u}$ in the system failure domain $\Omega_{F,\mathrm{sys}}=\{\boldsymbol{u} \in \mathbb{R}^n : \bigcup_{k=1}^K \bigcap_{i \in \boldsymbol{c}_k} ( G_i(\boldsymbol{u}) \leq 0 ) \}$.  

Next, the individual LSFs are linearized at an appropriately chosen linearization point $\boldsymbol{u}^*_i$.
For series system problems, the point $\boldsymbol{u}^*_i$ can be chosen as the most likely failure point (aka design point) of the $i$-th component reliability problem, which can be found through solving
\begin{equation} \label{eq:FORMopt}
\boldsymbol{u}^*_i = {\text{argmin}} \{\| \boldsymbol{u} \| \mid G_i(\boldsymbol{u}) = 0 \} \, .
\end{equation}
For parallel systems a better choice is the so-called joint design point, found through solving
\begin{equation} \label{eq:FORMoptPar}
\boldsymbol{u}^* = {\text{argmin}} \{\| \boldsymbol{u} \| \mid G_i(\boldsymbol{u}) \leq 0 , i=1,\ldots,m \} \, .
\end{equation} 
Assuming that each $G_i(\boldsymbol{u})$ is continuous and differentiable in the neighbourhood of the linearization point $\boldsymbol{u}^*_i$ (in case of parallel systems, the joint design point $\boldsymbol{u}^*$ should be used instead), one can approximate $G_i(\boldsymbol{u})$ in this neighbourhood through its first-order Taylor expansion around $\boldsymbol{u}^*_i$, which reads
\begin{equation} \label{eq:G1}
G_i(\boldsymbol{u}) \cong G_{1,i}(\boldsymbol{u}) = \nabla G_i(\boldsymbol{u}^*_i) (\boldsymbol{u}-\boldsymbol{u}^*_i) = \| \nabla G(\boldsymbol{u}^*_i) \| (\beta_i - \boldsymbol{\alpha}_i \boldsymbol{u} )\, .
\end{equation}
Here $\nabla G_i(\boldsymbol{u}^*_i) = [\partial G_i / \partial u_1 \rvert_{\boldsymbol{u} = \boldsymbol{u}^*_i}, \ldots, \partial G_i / \partial u_n \rvert_{\boldsymbol{u} = \boldsymbol{u}^*_i}]$ is the gradient row vector, $\boldsymbol{\alpha}_i = -\nabla G_i(\boldsymbol{u}^*_i)  /  \| \nabla G_i(\boldsymbol{u}^*_i) \|$ is the normalized negative gradient vector at $\boldsymbol{u}^*_i$ and $\beta_i = \boldsymbol{\alpha}_i \boldsymbol{u}^*_i$ is the FORM component reliability index.
The FORM approximation of the failure event $F_i$ is
\begin{equation} \label{eq:F1}
F_i \cong F_{1,i} = \{ G_{1,i}(\boldsymbol{U}) \leq 0\} = \{\boldsymbol{\alpha}_i \boldsymbol{U} \geq \beta_i \} \, ,
\end{equation}
and the corresponding approximation of the probability of the component failure event reads:
\begin{equation}
p_{F,i} \cong p_{F_{1,i}} = \Pr( \boldsymbol{\alpha}_i \boldsymbol{U} \geq \beta_i ) = \Phi (-\beta_i) \, .
\end{equation}
The components of the vector $\boldsymbol{\alpha}_i$, also known as ${\alpha}$-factors of the corresponding component failure event $F_i$, are often used in reliability analysis to assess the contribution of each random variable in $\boldsymbol{U}$ to the probability of component failure.
The connection of the ${\alpha}$-factors with the variance-based indices of Section~\ref{sec:VBrelSens} for component reliability problems has been shown in \cite{papaioannou2021variance}.
%

Define the random variables $Y_i = \boldsymbol{\alpha}_i \boldsymbol{U}, i=1,\ldots,m$.
The variables $Y_i, i=1,\ldots,m,$ have zero means, unit variances, correlation coefficients $\rho_{ij}=\boldsymbol{\alpha}_i \boldsymbol{\alpha}_j^{\mathrm{T}}$ and they follow the multivariate normal distribution.
Define the matrix $\boldsymbol{A}$  of dimension $m \times n$ with $i$-th row equal to $\boldsymbol{\alpha}_i$.
The correlation matrix of the random vector $[Y_1;\ldots;Y_m]$ is $\boldsymbol{R} = \boldsymbol{A}\boldsymbol{A}^\mathrm{T}$.
The probability of the series system problem of Eq.~\eqref{eq:Fser} is approximated by the probability of the event $F_{1,\mathrm{ser}}=\bigcup_{i=1}^m F_{1,i}$, which reads \cite{hohenbichler1982first}:
\begin{equation} \label{eq:PFserFORM}
p_{F,\mathrm{ser}} \cong p_{F_1,\mathrm{ser}} =\Pr \left( \bigcup_{i=1}^m F_{1,i} \right) =\Pr \left( \bigcup_{i=1}^m \{ Y_i \geq \beta_i \} \right)= 1-\Phi_m(\boldsymbol{B},\boldsymbol{R}) \, ,
\end{equation}
where $\Phi_m(\cdot,\boldsymbol{R})$ is the $m$-variate standard normal CDF with correlation matrix $\boldsymbol{R}$ and $\boldsymbol{B} = [\beta_1,\ldots,\beta_m]$.
Similarly, the probability of the parallel system problem of Eq.~\eqref{eq:Fpar} is approximated by the probability of the event $F_{1,\mathrm{par}}=\bigcap_{i=1}^m F_{1,i}$, which gives \citep{hohenbichler1982first}
\begin{equation} \label{eq:PFparFORM}
p_{F,\mathrm{par}} \cong p_{F_1,\mathrm{par}} =\Pr \left( \bigcap_{i=1}^m F_{1,i} \right) =\Pr \left( \bigcap_{i=1}^m \{ Y_i \geq \beta_i \} \right)= \Phi_m(-\boldsymbol{B},\boldsymbol{R}) \, .
\end{equation}
FORM approximations of general system problems can be obtained by expressing the system probability of the linearized components as a sum of parallel system probabilities through application of Poincar\'e's formula (aka inclusion-exclusion principle), which gives \cite{ADK05}:
\begin{equation} \label{eq:PFsysFORM}
p_{F,\mathrm{sys}} \cong p_{F_1,\mathrm{sys}} = \Pr \left( \bigcup_{k=1}^K C_{1,k} \right) \, = \sum_{\boldsymbol{v} \in \mathcal{P}(\{1,\ldots,K\}) \backslash \{ \emptyset \} } (-1)^{| \boldsymbol{v} |-1} \Pr \left( \bigcap_{i \in \boldsymbol{v}} C_{1,i} \right),
\end{equation}
where $C_{1,k} = \cap_{i \in \boldsymbol{c}_k} F_{1,i}$.
The parallel system probabilities in Eq.~\eqref{eq:PFsysFORM} are approximated through the FORM parallel system approximation of Eq.~\eqref{eq:PFparFORM}.
We note that accurate estimation of the general system problem requires linearization of the individual components entering each parallel system in Eq.~\eqref{eq:PFsysFORM} at the corresponding joint design point.

As seen in Eqs.~\eqref{eq:PFserFORM}--\eqref{eq:PFsysFORM}, the FORM approximations of system reliability problems require the evaluation of the multinormal probability integral, the estimation of which becomes cumbersome with increasing number of component failure events. 
Hence, system reliability bounds are often used, which require evaluation of joint probabilities of a small number of cut sets \cite{ditlevsen1979narrow,zhang1993high,song2003bounds}.

\begin{remark}
The FORM approximations of the series system problem can also be used to approximate the probability of failure of component problems with multiple design points, i.e., situations where Eq.~\eqref{eq:FORMopt} has multiple local minima, provided that all relevant design points can be identified efficiently, e.g., through application of the method discussed in \cite{der1998multiple}.
\end{remark}

\section{Variance-based reliability sensitivities of system problems with FORM} \label{s:sensitivities}

We derive expressions for the reliability sensitivity indices defined in Eq.~\eqref{eq:SobolRel} and Eq.~\eqref{eq:totalEffectRel} of the system reliability problems described by the linearized LSFs of Eq.~\eqref{eq:G1}.
Define the Bernoulli random variables $Z_{1,i} = \mathrm{I}_{F,i}(\boldsymbol{U} )$, describing the geometry of the failure domain of the FORM approximation of the component failure events in $\boldsymbol{U}$-space.
$Z_{1,i}$ has outcome 1 if $G_{1,i}(\boldsymbol{u}) \leq 0$ and 0 otherwise.
For the series problem, we define $Z_{1,\mathrm{ser}} = 1-\prod_{i=1}^m (1-Z_{1,i})$, and for the parallel problem we define $Z_{1,\mathrm{par}} = \prod_{i=1}^m Z_{1,i}$.
$Z_{1,\mathrm{ser}}$ and $Z_{1,\mathrm{par}}$ are Bernoulli random variables with parameters $p_{F_1,\mathrm{ser}}$ and $p_{F_1,\mathrm{par}}$, respectively.
The Sobol' indices $S_{F_1,\mathrm{ser},\boldsymbol{v}}$ and $S_{F_1,\mathrm{par},\boldsymbol{v}}$ of the FORM approximation of the series and parallel system failure events associated with index set $\boldsymbol{v}$ can be defined as
\begin{align}
S_{F_1,\mathrm{ser},\boldsymbol{v}} &=\frac{V_{F_1,\mathrm{ser},\boldsymbol{v}}}{\mathrm{Var}(Z_{1,\mathrm{ser}})}=\frac{V_{F_1,\mathrm{ser},\boldsymbol{v}}}{p_{F_{1,\mathrm{ser}}}(1-p_{F_{1,\mathrm{ser}}})} \, , \label{eq:SobolRel1ser} \\
S_{F_1,\mathrm{par},\boldsymbol{v}} &=\frac{V_{F_1,\mathrm{par},\boldsymbol{v}}}{\mathrm{Var}(Z_{1,\mathrm{par}})}=\frac{V_{F_1,\mathrm{par},\boldsymbol{v}}}{p_{F_{1,\mathrm{par}}}(1-p_{F_{1,\mathrm{par}}})} \, , \label{eq:SobolRel1par}
\end{align} 
with
\begin{align}
V_{F_1,\mathrm{ser},\boldsymbol{v}} & = \mathrm{Var}(\mathrm{E}[Z_{1,\mathrm{ser}} \vert \boldsymbol{U}_{\boldsymbol{v}} ]) - \sum_{\boldsymbol{w} \in \mathcal{P}(\boldsymbol{v}) \backslash \{\emptyset,\boldsymbol{v} \} } V_{F_1,\mathrm{ser},\boldsymbol{w}} \, ,  \label{eq:ViCondExpF1ser} \\
V_{F_1,\mathrm{par},\boldsymbol{v}} & = \mathrm{Var}(\mathrm{E}[Z_{1,\mathrm{par}} \vert \boldsymbol{U}_{\boldsymbol{v}} ]) - \sum_{\boldsymbol{w} \in \mathcal{P}(\boldsymbol{v}) \backslash \{\emptyset,\boldsymbol{v} \} } V_{F_1,\mathrm{par},\boldsymbol{w}} \, . \label{eq:ViCondExpF1par}
\end{align}
Evaluation of $S_{F_1,\mathrm{ser},\boldsymbol{v}}$ and $S_{F_1,\mathrm{par},\boldsymbol{v}}$ requires evaluation of the variances of conditional expectations $\mathrm{Var}(\mathrm{E}[Z_{1,\mathrm{ser}} \vert \boldsymbol{U}_{\boldsymbol{v}} ])$ and $\mathrm{Var}(\mathrm{E}[Z_{1,\mathrm{par}} \vert \boldsymbol{U}_{\boldsymbol{v}} ])$, for all $\boldsymbol{v} \in \mathcal{P}(\{1,\ldots,n\})$.
\begin{prop}
\label{prop:VarEZ1sys}
The variances of conditional expectations of the random variables $Z_{1,\mathrm{ser}}$ and $Z_{1,\mathrm{par}}$, $\mathrm{Var}(\mathrm{E}[Z_{1,\mathrm{ser}} \vert \boldsymbol{U}_{\boldsymbol{v}} ])$ and $\mathrm{Var}(\mathrm{E}[Z_{1,\mathrm{par}} \vert \boldsymbol{U}_{\boldsymbol{v}} ])$, with $\boldsymbol{U}_{\boldsymbol{v}}= \{U_i,i \in \boldsymbol{v} \}$ can be expressed through the following multivariate normal integrals:
\begin{equation} \label{eq:VarCondExpZ1ser}
\mathrm{Var}(\mathrm{E}[Z_{1,\mathrm{ser}} \vert \boldsymbol{U}_{\boldsymbol{v}} ]) = \Phi_{2m}\left(\begin{bmatrix} \boldsymbol{B} \\ \boldsymbol{B} \end{bmatrix},\begin{bmatrix} \boldsymbol{R} & \boldsymbol{R}_{\boldsymbol{v}} \\ \boldsymbol{R}_{\boldsymbol{v}} & \boldsymbol{R} \end{bmatrix} \right)-(1-p_{F_1,\mathrm{ser}})^2 \, ,
\end{equation}
and
\begin{equation} \label{eq:VarCondExpZ1par}
\mathrm{Var}(\mathrm{E}[Z_{1,\mathrm{par}} \vert \boldsymbol{U}_{\boldsymbol{v}} ]) = \Phi_{2m}\left(-\begin{bmatrix} \boldsymbol{B} \\ \boldsymbol{B} \end{bmatrix},\begin{bmatrix} \boldsymbol{R} & \boldsymbol{R}_{\boldsymbol{v}} \\ \boldsymbol{R}_{\boldsymbol{v}} & \boldsymbol{R} \end{bmatrix} \right)-p_{F_1,\mathrm{par}}^2 \, ,
\end{equation}
where $\boldsymbol{R}_{\boldsymbol{v}}=\boldsymbol{A}_{\boldsymbol{v}}\boldsymbol{A}_{\boldsymbol{v}}^{\mathrm{T}}$ and $\boldsymbol{A}_{\boldsymbol{v}}$ is a matrix of dimensions $m \times |\boldsymbol{v}|$ whose $i$-th row is equal to $\boldsymbol{\alpha}_{i,\boldsymbol{v}}=\{ \alpha_{i,j},j \in \boldsymbol{v} \}$, $\alpha_{i,j}$ being the $(i,j)$ element of matrix $\boldsymbol{A}$.
\end{prop}
\begin{proof}
We start with Eq.~\eqref{eq:VarCondExpZ1par}.
The variance $\mathrm{Var}(\mathrm{E}[Z_{1,\mathrm{par}} \vert \boldsymbol{U}_{\boldsymbol{v}} ])$ can be expanded as follows
\begin{equation} \label{eq:CondVarZ1par}
\mathrm{Var}(\mathrm{E}[Z_{1,\mathrm{par}} \vert \boldsymbol{U}_{\boldsymbol{v}} ]) = \mathrm{Var}(\Pr(F_{1,\mathrm{par}} \vert \boldsymbol{U}_{\boldsymbol{v}} )) = \mathrm{E}[\Pr(F_{1,\mathrm{par}} \vert \boldsymbol{U}_{\boldsymbol{v}} )^2 ] - \mathrm{E}[\Pr(F_{1,\mathrm{par}} \vert \boldsymbol{U}_{\boldsymbol{v}} ) ]^2 \, .
\end{equation}
The mean of the conditional probability of $F_{1,\mathrm{par}}$ is equal to the unconditional probability,
\begin{equation} \label{eq:ECondPF1par}
\mathrm{E}[\Pr(F_{1,\mathrm{par}} \vert \boldsymbol{U}_{\boldsymbol{v}} ) ] = p_{F_1,\mathrm{par}} \, .
\end{equation} 
The conditional probability of $F_{1,\mathrm{par}}$ given $\{ \boldsymbol{U}_{\boldsymbol{v}}=\boldsymbol{u}_{\boldsymbol{v}} \}$ reads:
\begin{equation}
\Pr(F_{1,\mathrm{par}} \vert \boldsymbol{U}_{\boldsymbol{v}} =\boldsymbol{u}_{\boldsymbol{v}} ) = \Pr \left( \bigcap_{i=1}^m \{ \boldsymbol{\alpha}_{i, \sim \boldsymbol{v}}  \boldsymbol{U}_{\sim \boldsymbol{v}} \geq \beta_i - \boldsymbol{\alpha}_{i,\boldsymbol{v}}  \boldsymbol{u}_{\boldsymbol{v}}\} \right) \, .
\end{equation}
The variables $\{\overline{Y}_i =  \boldsymbol{\alpha}_{i, \sim \boldsymbol{v}}  \boldsymbol{U}_{\sim \boldsymbol{v}},i=1,\ldots,m \}$ follow the multivariate normal distribution with zero mean and covariance matrix $\boldsymbol{R}_{\sim \boldsymbol{v}}=\boldsymbol{A}_{\sim \boldsymbol{v}}\boldsymbol{A}_{\sim \boldsymbol{v}}^{\mathrm{T}}=\boldsymbol{R}-\boldsymbol{R}_{\boldsymbol{v}}$. Therefore:
\begin{equation}
\Pr(F_{1,\mathrm{par}} \vert \boldsymbol{U}_{\boldsymbol{v}} =  \boldsymbol{u}_{\boldsymbol{v}} ) = \Pr \left( \bigcap_{i=1}^m \{ \overline{Y}_i \geq \beta_i - \boldsymbol{\alpha}_{i,\boldsymbol{v}}  \boldsymbol{u}_{\boldsymbol{v}}\} \right)  = \Phi_m(\boldsymbol{A}_{ \boldsymbol{v}}\boldsymbol{u}_{\boldsymbol{v}}-\boldsymbol{B},\boldsymbol{R}-\boldsymbol{R}_{\boldsymbol{v}}) \, .
\end{equation} 
We then have:
\begin{equation} \label{eq:ECondPF1par2der}
\begin{aligned}
&\mathrm{E}[\Pr(F_{1,\mathrm{par}} \vert \boldsymbol{U}_{\boldsymbol{v}} )^2 ] = \mathrm{E}[\Phi_m(\boldsymbol{A}_{ \boldsymbol{v}}\boldsymbol{U}_{\boldsymbol{v}}-\boldsymbol{B},\boldsymbol{R}-\boldsymbol{R}_{\boldsymbol{v}})^2 ] \\
& = \mathrm{E} \left[\Pr \left( \left.\widetilde{\boldsymbol{U}}_1 \leq \boldsymbol{A}_{\boldsymbol{v}}\boldsymbol{U}_{\boldsymbol{v}}-\boldsymbol{B} \right\vert \boldsymbol{U}_{\boldsymbol{v}} \right) \Pr \left( \left.\widetilde{\boldsymbol{U}}_2 \leq \boldsymbol{A}_{\boldsymbol{v}}\boldsymbol{U}_{\boldsymbol{v}}-\boldsymbol{B} \right\vert \boldsymbol{U}_{\boldsymbol{v}} \right) \right] \\
& = \mathrm{E} \left[ \Pr\left( \left. \left\{ \widetilde{\boldsymbol{U}}_1 \leq \boldsymbol{A}_{\boldsymbol{v}}\boldsymbol{U}_{\boldsymbol{v}}-\boldsymbol{B} \right\} \cap \left\{ \widetilde{\boldsymbol{U}}_2 \leq \boldsymbol{A}_{\boldsymbol{v}}\boldsymbol{U}_{\boldsymbol{v}}-\boldsymbol{B} \right\} \right\vert \boldsymbol{U}_{\boldsymbol{v}}   \right)  \right]\\
&=\Pr\left( \left\{ \widetilde{\boldsymbol{U}}_1 \leq \boldsymbol{A}_{ \boldsymbol{v}}\boldsymbol{U}_{\boldsymbol{v}}-\boldsymbol{B} \right\} \cap  \left\{ \widetilde{\boldsymbol{U}}_2 \leq \boldsymbol{A}_{ \boldsymbol{v}}\boldsymbol{U}_{\boldsymbol{v}}-\boldsymbol{B} \right\}  \right) \, ,
\end{aligned}
\end{equation}
where $\widetilde{\boldsymbol{U}}_1$ and $\widetilde{\boldsymbol{U}}_2$ are auxiliary random vectors with zero means and covariance matrices $\boldsymbol{R}-\boldsymbol{R}_{\boldsymbol{v}}$ and each component of $\widetilde{\boldsymbol{U}}_1$ is independent from all components of $\widetilde{\boldsymbol{U}}_2$.
Define the random vectors $\widetilde{\boldsymbol{Y}}_i=\widetilde{\boldsymbol{U}}_i - \boldsymbol{A}_{ \boldsymbol{v}}\boldsymbol{U}_{ \boldsymbol{v}}, i=1,2$. 
Each random vector $\widetilde{\boldsymbol{Y}}_i,i=1,2,$ has zero mean and covariance matrix $\boldsymbol{R}$, while they have cross-covariance matrix $\boldsymbol{R}_{\boldsymbol{v}}$ and, since they are linear functions of normal random variables, they follow the $2m$-variate standard normal distribution. We have:
\begin{equation} \label{eq:ECondPF1par2}
\mathrm{E}[\Pr(F_{1,\mathrm{par}} \vert \boldsymbol{U}_{\boldsymbol{v}} )^2 ] = \Pr\left( \left\{ \widetilde{\boldsymbol{Y}}_1 \leq -\boldsymbol{B} \right\} \cap \left\{ \widetilde{\boldsymbol{Y}}_2 \leq -\boldsymbol{B}\right\} \right) = \Phi_{2m}\left(-\begin{bmatrix} \boldsymbol{B} \\ \boldsymbol{B} \end{bmatrix},\begin{bmatrix} \boldsymbol{R} & \boldsymbol{R}_{\boldsymbol{v}} \\ \boldsymbol{R}_{\boldsymbol{v}} & \boldsymbol{R} \end{bmatrix} \right) \, .
\end{equation}
Combining Eqs.~\eqref{eq:ECondPF1par2}, \eqref{eq:CondVarZ1par} and \eqref{eq:ECondPF1par} we get the result of Eq.~\eqref{eq:VarCondExpZ1par}.

In the same way as Eq.~\eqref{eq:CondVarZ1par}, we can decompose $\mathrm{Var}(\mathrm{E}[Z_{1,\mathrm{ser}} \vert \boldsymbol{U}_{\boldsymbol{v}} ])$ as follows
\begin{equation} \label{eq:CondVarZ1ser}
\mathrm{Var}(\mathrm{E}[Z_{1,\mathrm{ser}} \vert \boldsymbol{U}_{\boldsymbol{v}} ]) = \mathrm{E}[\Pr(F_{1,\mathrm{ser}} \vert \boldsymbol{U}_{\boldsymbol{v}} )^2 ] - p_{F_1,\mathrm{ser}}^2\, .
\end{equation}
It is:
\begin{equation} \label{eq:CondPF1ser}
\begin{aligned}
\Pr(F_{1,\mathrm{ser}} \vert \boldsymbol{U}_{\boldsymbol{v}} =\boldsymbol{u}_{\boldsymbol{v}} ) & = \Pr \left( \bigcup_{i=1}^m \{ \boldsymbol{\alpha}_{i, \sim \boldsymbol{v}}  \boldsymbol{U}_{\sim \boldsymbol{v}} \geq \beta_i - \boldsymbol{\alpha}_{i,\boldsymbol{v}}  \boldsymbol{u}_{\boldsymbol{v}}\} \right) \\
&= 1- \Pr \left( \bigcap_{i=1}^m \{ \boldsymbol{\alpha}_{i, \sim \boldsymbol{v}}  \boldsymbol{U}_{\sim \boldsymbol{v}} < \beta_i - \boldsymbol{\alpha}_{i,\boldsymbol{v}}  \boldsymbol{u}_{\boldsymbol{v}}\} \right) \\
&= 1- \Pr \left( \bigcap_{i=1}^m \{ \overline{Y}_i< \beta_i - \boldsymbol{\alpha}_{i,\boldsymbol{v}}  \boldsymbol{u}_{\boldsymbol{v}}\} \right) \\
& = 1- \Phi_m(\boldsymbol{B}-\boldsymbol{A}_{ \boldsymbol{v}}\boldsymbol{u}_{\boldsymbol{v}},\boldsymbol{R}-\boldsymbol{R}_{\boldsymbol{v}})  \, .
\end{aligned}
\end{equation}
Combining Eqs.~\eqref{eq:CondPF1ser} and \eqref{eq:CondVarZ1ser}, we get
\begin{equation} \label{eq:CondVarZ1ser2}
\begin{aligned}
\mathrm{Var}(\mathrm{E}[Z_{1,\mathrm{ser}} \vert \boldsymbol{U}_{\boldsymbol{v}} ]) & = 1 - 2(1-p_{F_1,\mathrm{ser}})+ \mathrm{E}[\Phi_m(\boldsymbol{B}-\boldsymbol{A}_{ \boldsymbol{v}}\boldsymbol{u}_{\boldsymbol{v}},\boldsymbol{R}-\boldsymbol{R}_{\boldsymbol{v}})^2 ] - p_{F_1,\mathrm{ser}}^2 \\ 
&=  \mathrm{E}[\Phi_m(\boldsymbol{B}-\boldsymbol{A}_{ \boldsymbol{v}}\boldsymbol{u}_{\boldsymbol{v}},\boldsymbol{R}-\boldsymbol{R}_{\boldsymbol{v}})^2 ] - (1- p_{F_1,\mathrm{ser}})^2 \, .
\end{aligned}
\end{equation}
Following the same approach as in Eqs.~\eqref{eq:ECondPF1par2der} and \eqref{eq:ECondPF1par2}, we get
\begin{equation} \label{eq:EPhimser}
\mathrm{E}[\Phi_m(\boldsymbol{B}-\boldsymbol{A}_{ \boldsymbol{v}}\boldsymbol{u}_{\boldsymbol{v}},\boldsymbol{R}-\boldsymbol{R}_{\boldsymbol{v}})^2 ]  = \Phi_{2m}\left(\begin{bmatrix} \boldsymbol{B} \\ \boldsymbol{B} \end{bmatrix},\begin{bmatrix} \boldsymbol{R} & \boldsymbol{R}_{\boldsymbol{v}} \\ \boldsymbol{R}_{\boldsymbol{v}} & \boldsymbol{R} \end{bmatrix} \right) \, .
\end{equation}
Combining Eqs.~\eqref{eq:EPhimser} and \eqref{eq:CondVarZ1ser2} we get the result of Eq.~\eqref{eq:VarCondExpZ1ser}.
\end{proof}

The total-effect indices of $Z_{1,\mathrm{ser}}$ and $Z_{1,\mathrm{par}}$ associated with component indices ${\boldsymbol{v}}$ are given by Eq.~\eqref{eq:totalEffectRel} as
\begin{equation} \label{eq:totalEffectRel1ser}
S_{F_1,\boldsymbol{v}}^T = 1- \frac{\mathrm{Var}(\mathrm{E}[Z_{1,\mathrm{ser}} \vert \boldsymbol{U}_{\sim \boldsymbol{v}} ]) }{\mathrm{Var}(Z_{1,\mathrm{ser}})} = 1- \frac{\mathrm{Var}(\Pr(F_{1,\mathrm{ser}} \vert \boldsymbol{U}_{\sim \boldsymbol{v}} )) }{p_{F_1,\mathrm{ser}}(1-p_{F_1,\mathrm{ser}})} \, ,
\end{equation}
and
\begin{equation} \label{eq:totalEffectRel1par}
S_{F_1,\boldsymbol{v}}^T = 1- \frac{\mathrm{Var}(\mathrm{E}[Z_{1,\mathrm{par}} \vert \boldsymbol{U}_{\sim \boldsymbol{v}} ]) }{\mathrm{Var}(Z_{1,\mathrm{par}})} = 1- \frac{\mathrm{Var}(\Pr(F_{1,\mathrm{par}} \vert \boldsymbol{U}_{\sim \boldsymbol{v}} )) }{p_{F_1,\mathrm{par}}(1-p_{F_1,\mathrm{par}})} \, .
\end{equation}
Evaluation of $S_{F_1,\mathrm{ser},\boldsymbol{v}}^T$ and $S_{F_1,\mathrm{par},\boldsymbol{v}}^T$ requires evaluation of $\mathrm{Var}(\Pr(F_{1,\mathrm{ser}} \vert \boldsymbol{U}_{\sim \boldsymbol{v}} )) $ and $\mathrm{Var}(\Pr(F_{1,\mathrm{par}} \vert \boldsymbol{U}_{\sim \boldsymbol{v}} )) $.
Applying Proposition~\ref{prop:VarEZ1sys} and the fact that $\boldsymbol{R}_{\sim \boldsymbol{v}}=\boldsymbol{R}-\boldsymbol{R}_{\boldsymbol{v}}$ gives
\begin{equation} \label{eq:VarCondExpZ1Tser}
\mathrm{Var}(\Pr(F_{1,\mathrm{ser}} \vert \boldsymbol{U}_{\sim \boldsymbol{v}} )) = \Phi_{2m}\left(\begin{bmatrix} \boldsymbol{B} \\ \boldsymbol{B} \end{bmatrix},\begin{bmatrix} \boldsymbol{R} & \boldsymbol{R} - \boldsymbol{R}_{\boldsymbol{v}} \\ \boldsymbol{R}- \boldsymbol{R}_{\boldsymbol{v}} & \boldsymbol{R} \end{bmatrix} \right) -(1-p_{F_1,\mathrm{ser}})^2 \, ,
\end{equation}
and 
\begin{equation}
\mathrm{Var}(\Pr(F_{1,\mathrm{par}} \vert \boldsymbol{U}_{\sim \boldsymbol{v}} )) = \Phi_{2m}\left(-\begin{bmatrix} \boldsymbol{B} \\ \boldsymbol{B} \end{bmatrix},\begin{bmatrix} \boldsymbol{R} & \boldsymbol{R} - \boldsymbol{R}_{\boldsymbol{v}} \\ \boldsymbol{R}- \boldsymbol{R}_{\boldsymbol{v}} & \boldsymbol{R} \end{bmatrix} \right) -p_{F_1,\mathrm{par}}^2 \, .
\end{equation}

\begin{remark}
The first-order and total-effect reliability sensitivity indices of the series and parallel system problems associated with the $i$-th random variable can be obtained through substituting $\boldsymbol{v} = \{i\}$ in Eqs.~\eqref{eq:SobolRel1ser}-\eqref{eq:SobolRel1par}, and Eqs.~\eqref{eq:totalEffectRel1ser}-\eqref{eq:totalEffectRel1par}.
In this case, the correlation matrix $\boldsymbol{R}_i$ has $(k,l)$ entry $\alpha_{k,i}\alpha_{l,i}$, with $\alpha_{k,i}$ denoting the $i$-th entry of vector $\boldsymbol{\alpha}_k$.
\end{remark}

\begin{remark}
Dividing Eqs. \eqref{eq:VarCondExpZ1ser} and \eqref{eq:VarCondExpZ1par} with $\mathrm{Var}(Z_{1,\mathrm{ser}})$ or $\mathrm{Var}(Z_{1,\mathrm{par}})$ directly results in the closed Sobol' index $S_{F_1,\boldsymbol{v}}^{clo}$ for the series or parallel system.
\end{remark}

The results for the series and parallel system problems can be extended to general systems by application of Eq.~\eqref{eq:PFsysFORM}.
Consider again the definition of the general system of Eq.~\eqref{eq:Fgensys} and let $Z_{1,\mathrm{sys}} = 1-\prod_{k=1}^K (1-\prod_{i\in \boldsymbol{c}_k} Z_{1,i})$.
The Sobol' index $S_{F_1,\mathrm{sys},\boldsymbol{v}}$ and total-effect index $S_{F_1,\mathrm{sys},\boldsymbol{v}}^T$ associated with the set $\boldsymbol{v}$ are given as
\begin{equation} \label{eq:Sobol1ser}
S_{F_1,\mathrm{sys},\boldsymbol{v}} = \frac{V_{F_1,\mathrm{sys},\boldsymbol{v}}}{p_{F_1,\mathrm{sys}}(1-p_{F_1,\mathrm{sys}})} \, , \qquad S_{F_1,\mathrm{sys},\boldsymbol{v}}^T = 1- \frac{\mathrm{Var}(\mathrm{E}[Z_{1,\mathrm{sys}} \vert \boldsymbol{U}_{\sim \boldsymbol{v}} ])  }{p_{F_1,\mathrm{sys}}(1-p_{F_1,\mathrm{sys}})} \, ,
\end{equation}
with $V_{F_1,\mathrm{sys},\boldsymbol{v}} = \mathrm{Var}(\mathrm{E}[Z_{1,\mathrm{sys}} \vert \boldsymbol{U}_{\boldsymbol{v}} ]) - \sum_{\boldsymbol{w} \in \mathcal{P}(\boldsymbol{v}) \backslash \{\emptyset,\boldsymbol{v} \} } V_{F_1,\mathrm{psys},\boldsymbol{w}}$.
The variance of conditional expectation of the general system is obtained as
\begin{equation}
 \mathrm{Var}(\mathrm{E}[Z_{1,\mathrm{sys}} \vert \boldsymbol{U}_{\boldsymbol{v}} ])  =\mathrm{E}[\Pr(F_{1,\mathrm{sys}} \vert \boldsymbol{U}_{\boldsymbol{v}} )^2 ] - p_{F_1,\mathrm{sys}}^2\, ,
\end{equation}
where
\begin{equation} \label{eq:ECondPF1sys}
\footnotesize
\begin{aligned}
 & \mathrm{E}[\Pr(F_{1,\mathrm{sys}} \vert \boldsymbol{U}_{\boldsymbol{v}} )^2 ] =\\
 & \sum_{\boldsymbol{z} \in \mathcal{P}(\{1,\ldots,K\}) \backslash \{ \emptyset \} }  \sum_{\boldsymbol{w} \in \mathcal{P}(\{1,\ldots,K\}) \backslash \{ \emptyset \} } (-1)^{| \boldsymbol{z} | +| \boldsymbol{w} | -2} \Phi_{| \cap_{k \in \boldsymbol{z}} \boldsymbol{c}_k |  +| \cap_{k \in \boldsymbol{w}} \boldsymbol{c}_k | }\left(-\begin{bmatrix} \boldsymbol{B}_{\boldsymbol{z}} \\ \boldsymbol{B}_{\boldsymbol{w}} \end{bmatrix},\begin{bmatrix} \boldsymbol{R}_{\boldsymbol{z}} & \boldsymbol{R}_{\boldsymbol{z}\boldsymbol{w},\boldsymbol{v}} \\ \boldsymbol{R}_{\boldsymbol{z}\boldsymbol{w},\boldsymbol{v}}^\mathrm{T}  & \boldsymbol{R}_{\boldsymbol{w}} \end{bmatrix} \right) \, .
 \end{aligned}
\end{equation}
It is $\boldsymbol{B}_{\boldsymbol{z}} = [\beta_i,i\in \cap_{k \in \boldsymbol{z}} \boldsymbol{c}_k]$, $\boldsymbol{R}_{\boldsymbol{z}} =  \boldsymbol{A}_{\boldsymbol{z}} \boldsymbol{A}_{\boldsymbol{z}}^\mathrm{T} $ with $ \boldsymbol{A}_{\boldsymbol{z}}$ being a matrix of dimensions $|\cap_{k \in \boldsymbol{z}} \boldsymbol{c}_k | \times n$ whose $i$-th row is equal to $\boldsymbol{\alpha}_i  ,i\in \cap_{k \in \boldsymbol{z}} \boldsymbol{c}_k $, and $\boldsymbol{R}_{\boldsymbol{z}\boldsymbol{w},\boldsymbol{v}} =  \boldsymbol{A}_{\boldsymbol{z},\boldsymbol{v}} \boldsymbol{A}_{\boldsymbol{w},\boldsymbol{v}}^\mathrm{T} $.
The result of Eq.~\eqref{eq:ECondPF1sys} is obtained through expanding the conditional probability inside the expectation by use of the inclusion-exclusion principle of Eq.~\eqref{eq:PFsysFORM} and applying the same technique as in Eqs.~\eqref{eq:ECondPF1par2der}-\eqref{eq:ECondPF1par2} to each resulting summand.
Similarly, for the total effect indices, we have 
\begin{equation}
 \mathrm{Var}(\mathrm{E}[Z_{1,\mathrm{sys}} \vert \boldsymbol{U}_{\sim \boldsymbol{v}} ])  =\mathrm{E}[\Pr(F_{1,\mathrm{sys}} \vert \boldsymbol{U}_{\sim \boldsymbol{v}} )^2 ] - p_{F_1,\mathrm{sys}}^2\, ,
\end{equation}
where
\begin{equation}  \label{eq:ECondPFTsys}
\scriptsize
\begin{aligned}
 & \mathrm{E}[\Pr(F_{1,\mathrm{sys}} \vert \boldsymbol{U}_{\sim\boldsymbol{v}} )^2 ] =\\
 & \sum_{\boldsymbol{z} \in \mathcal{P}(\{1,\ldots,K\}) \backslash \{ \emptyset \} }  \sum_{\boldsymbol{w} \in \mathcal{P}(\{1,\ldots,K\}) \backslash \{ \emptyset \} } (-1)^{| \boldsymbol{z} | +| \boldsymbol{w} | -2} \Phi_{| \cap_{k \in \boldsymbol{z}} \boldsymbol{c}_k |  +| \cap_{k \in \boldsymbol{w}} \boldsymbol{c}_k | }\left(-\begin{bmatrix} \boldsymbol{B}_{\boldsymbol{z}} \\ \boldsymbol{B}_{\boldsymbol{w}} \end{bmatrix},\begin{bmatrix} \boldsymbol{R}_{\boldsymbol{z}} & \boldsymbol{R}_{\boldsymbol{z}\boldsymbol{w}}-\boldsymbol{R}_{\boldsymbol{z}\boldsymbol{w},\boldsymbol{v}} \\ (\boldsymbol{R}_{\boldsymbol{z}\boldsymbol{w}}-\boldsymbol{R}_{\boldsymbol{z}\boldsymbol{w},\boldsymbol{v}} )^\mathrm{T}  & \boldsymbol{R}_{\boldsymbol{w}} \end{bmatrix} \right) \, .
 \end{aligned}
\end{equation}
The detailed proofs of Eqs.~\eqref{eq:ECondPF1sys} and \eqref{eq:ECondPFTsys} are omitted for the sake of brevity.

\begin{remark}
Following Remark 4.1 in \cite{papaioannou2021variance}, the results presented here correspond to the Sobol’ and total-effect indices of $Z_{1,\mathrm{ser}}$, $Z_{1,\mathrm{par}}$ and $Z_{1,\mathrm{sys}}$ with respect to the original variables $\boldsymbol{X}$ for the case where $\boldsymbol{X}$ consists of statistically independent components.
That is, the presented results can be used to approximate the sensitivity indices of general system problems with independent inputs.
\end{remark}

\begin{remark}
The derived FORM approximations of the reliability sensitivity indices require the evaluation of multinormal integrals of dimension twice the number of components.
In problems with many components, the accurate and efficient evaluation of these integrals necessitates the use of tailored algorithms \cite[e.g.,][]{ambartzumian1998multinormal,botev2017normal,gessner2020integrals}.
In the numerical examples, we estimate the multinormal integrals using the Matlab function provided by the author of \cite{botev2017normal} and, if not otherwise noted, use $5 \times 10^7$ samples.
\end{remark}

\section{Numerical examples} \label{s:examples}

\subsection{Illustrative example}

The first example examines the behaviour of the reliability sensitivity indices in a two-dimensional and two-component system with LSFs given by linear functions of independent standard normal variables $\boldsymbol{U} = [U_1;U_2]$.
The components have LSFs,
\begin{equation} \nonumber
G_{i}(\boldsymbol{U}) = \beta_i - \boldsymbol{\alpha}_i \boldsymbol{U}, \qquad i=1,2 \, ,
\end{equation}
with $\boldsymbol{\alpha}_1 = [1/ \sqrt{2}, 1/ \sqrt{2}]$ and we vary the angle $\theta$ between $\boldsymbol{\alpha}_1$ and $\boldsymbol{\alpha}_2$, as illustrated in Fig~\ref{f:LinearSysIllustration}.
That is, $\boldsymbol{\alpha}_2 = [1/ \sqrt{2} \cos \theta +  1/ \sqrt{2} \sin \theta, 1/ \sqrt{2} \cos \theta -  1/ \sqrt{2} \sin \theta]$.
For each $\theta$, we evaluate the first-order and total-effect indices of the series and parallel systems for random variable $U_2$.
As the two component LSFs are linear in $\boldsymbol{U}$, we can evaluate the sensitivities by application of the expressions derived in Section~\ref{s:sensitivities}.
We also evaluate the sensitivities of the two component problems through application of the expressions derived in \cite{papaioannou2021variance}.

We first consider the case where $\beta_1 = \beta_2 = 2$.
The first-order and total-effect indices of the parallel, series and component problems for random variable $U_2$ are shown in Fig.~\ref{f:LinearSysS1ST1} for varying $\theta$.
We see that for $\theta = 0 ^{\circ}$ both the first-order and total-effect indices of the system problems are identical to the component indices.
The sensitivity indices of component 2 reach their maximum at $\theta = 45 ^{\circ}$, where $G_2$ becomes parallel to the horizontal axis and $U_2$ dominates the problem, and their minimum at $\theta = 135 ^{\circ}$, where $G_2$ is parallel to the vertical axis and $U_2$ has no influence.
The sensitivities of the series system problem have a similar behaviour to the ones of component 2, albeit with decreased magnitude due to the influence of component 1, which leads to a decrease in the importance of variable $U_2$.
The parallel system problem exhibits different behaviour.
The first-order index increases as $\theta$ increases until $\theta \approx 70 ^{\circ}$, where the influence of $U_2$ on the joint failure domain becomes maximum, and decreases again to a minimum value of $\approx 0$ as the joint failure domain becomes negligible.
The total effect index increases monotonically, until it reaches a value close to 1 for $\theta > 50 ^{\circ}$.
This behaviour is in agreement with the findings in \cite{papaioannou2021variance}, where it is shown that as the probability of failure of component problems becomes smaller, the first-order index decreases and total-effect index increases since the variance of the indicator function of the failure domain is dominated by high-order effects.

\begin{figure}[t!]
  \centering
    \includegraphics[width=\textwidth]{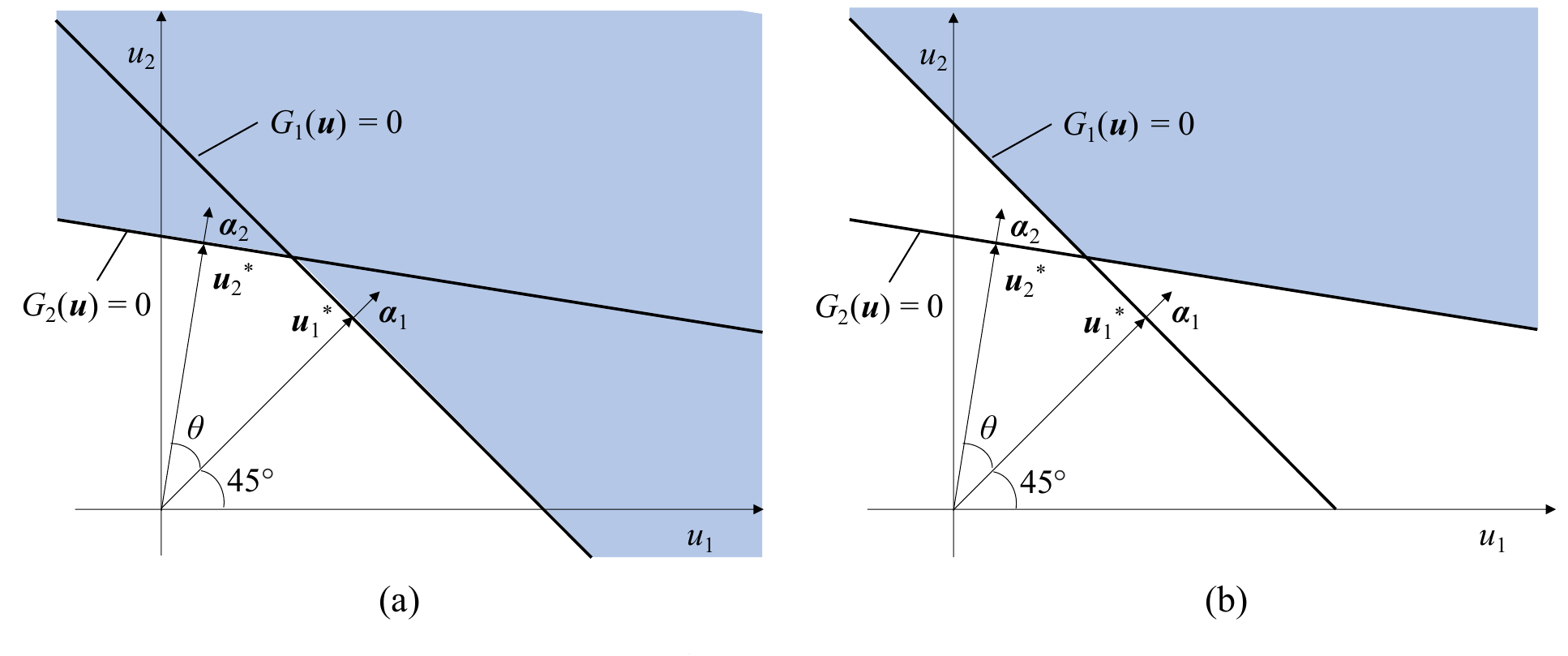}
    \caption{Illustration of the failure domains (shaded) of the (a) series and (b) parallel system problem with two linear component LSFs. Component 1 has $\boldsymbol{\alpha}_1 = [1/ \sqrt{2}, 1/ \sqrt{2}]$ and $\theta$ denotes the angle between $\boldsymbol{\alpha}_1$ and $\boldsymbol{\alpha}_2$.}
\label{f:LinearSysIllustration}
\end{figure}

\begin{figure}[t!]
  \centering
    \includegraphics[width=\textwidth]{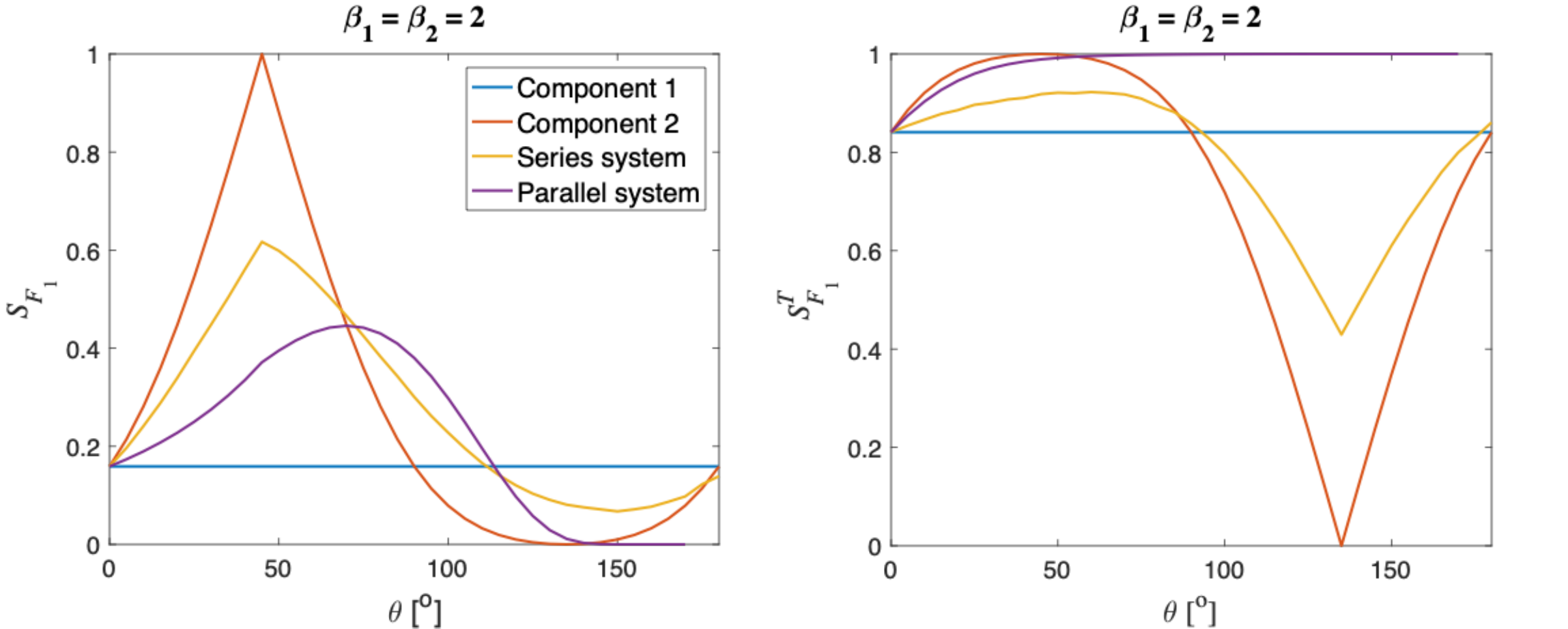}
    \caption{First-order and total-effect indices of input $U_2$ for the parallel and series system problems with linear LSFs as a function of the angle $\theta$ between them.}
\label{f:LinearSysS1ST1}
\end{figure}

Next, we look at the case where one of the two components has dominant contribution to the failure domain.
The first-order index of $U_2$ for the case where component 1 is dominant ($\beta_1 = 2, \beta_2 = 3$) and the case where component 2 is dominant ($\beta_1 = 3, \beta_2 = 2$) are shown is Fig.~\ref{f:LinearSysS1ST2}.
In both cases, the series system is governed by the dominant component, hence, the first-order index of the series system behaves similar to the index of component 1 in case 1 and component 2 in case 2.
Conversely, the parallel system is governed by the minor component, and the first-order index behaves like the index of component 2 in case 1 and component 1 in case 2. 
In both cases, the first-order index of the parallel system problem takes values close to zero for large $\theta$, which implies negligible probability mass in the failure domain.

\begin{figure}[t!]
  \centering
    \includegraphics[width=\textwidth]{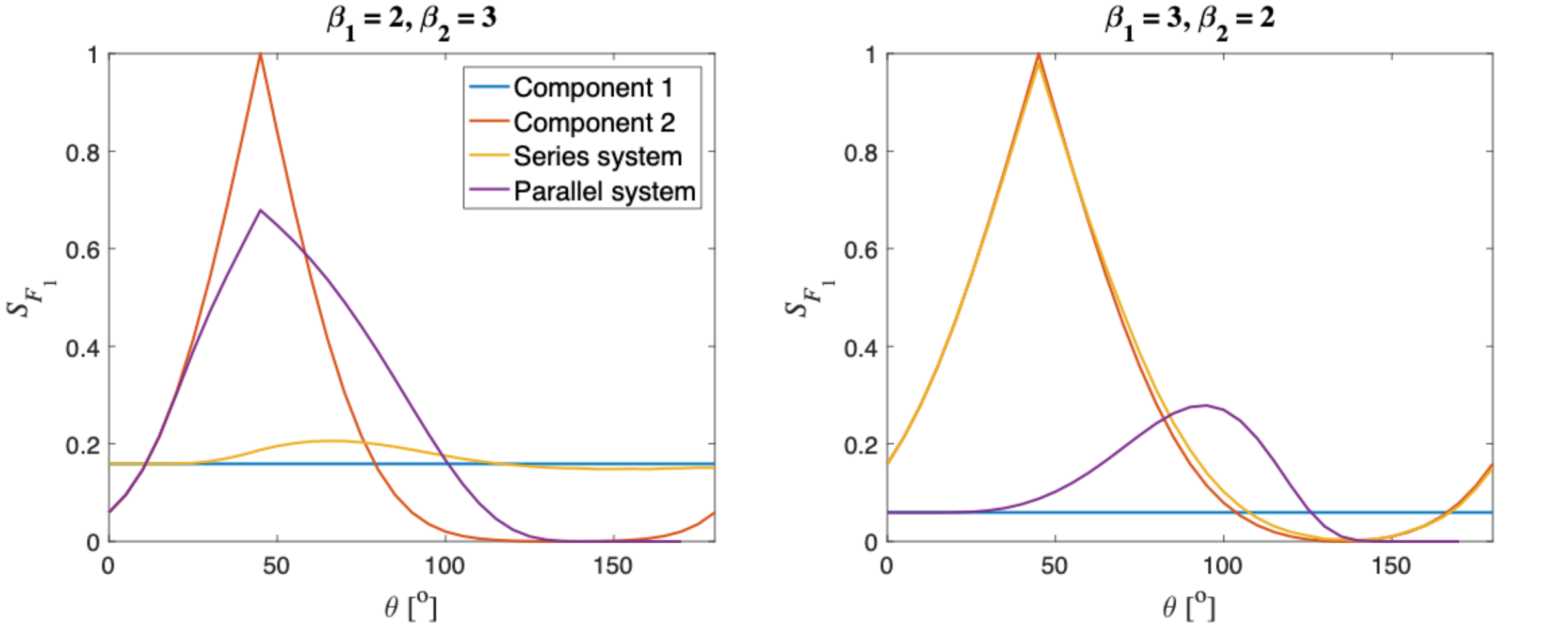}
    \caption{First-order indices of input $U_2$ for the parallel and series system problems with linear LSFs as a function of the angle $\theta$ between them. Left: Component 1 is dominant. Right: Component 2 is dominant.}
\label{f:LinearSysS1ST2}
\end{figure}

\subsection{Series system: Elastoplastic frame} \label{sec:serFORMExample}
We consider an elastoplastic frame structure subjected to a horizontal load, as shown in Fig.~\ref{f:frame}. This example is taken from \cite{zhao2003system}. The frame has four potential failure modes represented by the following LSFs:
\begin{eqnarray}  \nonumber
g_{1}(M_1,M_2,M_3,S) &=& 2M_1+2M_3-4.5S \, , \\ \nonumber
g_{2}(M_1,M_2,M_3,S) &=& 2M_1+M_2+M_3-4.5S \, , \\ \nonumber
g_{3}(M_1,M_2,M_3,S) &=& M_1+M_2+2M_3-4.5S \, , \\ \nonumber
g_{4}(M_1,M_2,M_3,S) &=& M_1+2M_2+M_3-4.5S \, . 
\end{eqnarray}
The system failure event is a series system with $i$-th component failure mode represented by LSF $g_i$.
The plastic moments $M_i, i=1,2,3$, and the load $S$ are modeled by independent lognormal random variables with means and standard deviations given in Table~\ref{t:modelex2}.

\begin{figure}[h]
  \centering
    \includegraphics[width=0.45\textwidth]{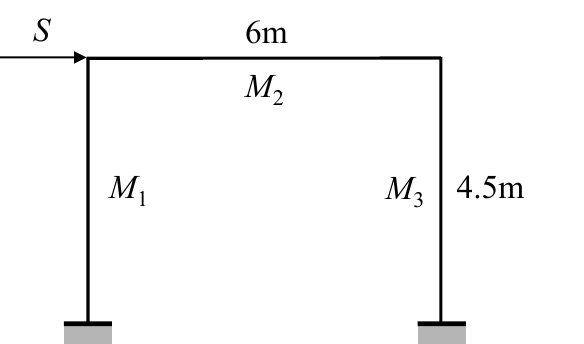}
    \caption{Elastoplastic frame according to \cite{zhao2003system}.}
\label{f:frame}
\end{figure}

\begin{table}
\caption{Uncertain parameters of the frame example.}
\footnotesize
\begin{center}
\begin{tabular}{llll}\hline
\noalign{\smallskip}
Parameter&Distribution&Mean&St. Dev.\\
\noalign{\smallskip}\hline
\noalign{\smallskip}
$M_i$ $(i=1,2,3)$ [tm]&Lognormal&$200$&$30$\\
$S$ [t]&Lognormal&$50$&$20$\\
\noalign{\smallskip}\hline
\end{tabular}
\end{center}
\label{t:modelex2}
\end{table}

The reference solution of the probability of failure of the system problem is evaluated with Monte Carlo and $n_s = 10^8$ samples as $p_{F,MC}= 5.32 \times 10^{-4}$.
Monte Carlo estimates of the sensitivity indices are obtained using the classical pick-freeze estimators (e.g., see \cite{prieur2017variance}), which require $n_s ( n+2)$ LSF evaluations.
The FORM estimate of the series system probability, obtained using the individual design points of each component problem, is $p_{F_1} = 5.57 \times 10^{-4}$.
The first-order and total-effect indices estimated with Monte Carlo and FORM are given in Table~\ref{t:resultsex2}.
The FORM estimates of the total-effect indices show excellent agreement with the Monte Carlo reference solution.
The first-order index of the load variable is in agreement with the Monte Carlo reference, whereas the indices of the plastic moments overestimate the reference values, although the absolute error of all first-order indices is comparable.
We note that the FORM results for this example depend strongly on the setting of the algorithm used to estimate the multivariate normal integrals \cite{botev2017normal}; particularly, $10^8$ samples was necessary to obtain a convergent solution.
This is likely due to the strong correlation between the individual component failure modes -- the off-diagonal entries of the matrix $\boldsymbol{R}$ range between $0.975$ and $0.992$ -- which leads to badly conditioned correlation matrices in Eqs.~\eqref{eq:VarCondExpZ1ser} and \eqref{eq:VarCondExpZ1Tser}.

\begin{table}
\caption{Estimates of the first-order and total-effect indices for the frame example.}
\footnotesize
\begin{center}
\begin{tabular}{llllll}\hline
\noalign{\smallskip}
Method& Index &$M_1$ & $M_2$ & $M_3$ & $S$\\
\noalign{\smallskip}\hline
\noalign{\smallskip}
MC & $S_{F,i}$ &$9.39 \times 10^{-5}$ & $1.69 \times 10^{-4}$ & $2.81 \times 10^{-4}$ & $0.558$ \\
& $S^T_{F,i}$ &$0.293$ & $0.150$ & $0.293$ & $1.00$ \\
FORM & $S_{F,i}$ &$4.75 \times 10^{-4}$ & $9.04 \times 10^{-4}$ & $7.59 \times 10^{-4}$ & $0.557$  \\
& $S^T_{F,i}$ &$0.288$ & $0.160$ & $0.281$ & $1.00$   \\
\noalign{\smallskip}\hline
\end{tabular}
\end{center}
\label{t:resultsex1}
\end{table}

\subsection{Parabolic limit state with two design points}

This example, taken from \cite{der1998multiple}, considers a component reliability problem with two design points.
The failure event is described by the following LSF
\begin{equation} \nonumber
G(\boldsymbol{U}) = b-U_2-\kappa (U_1-e)^2 \, ,
\end{equation}
where $b$, $\kappa$ and $e$ are parameters chosen as $b=5$, $\kappa=0.5$ and $e=0.1$.
The two design points are $\boldsymbol{u}^*_1 = [-2.741; 0.965]$ and $\boldsymbol{u}^*_2 = [2.916 , 1.036]$ \cite{der1998multiple}.
The limit state surface together with the two design points and the corresponding linearized limit states are shown in Fig.~\ref{f:parabolicLSF}.
The reference solution of the probability of failure is evaluated with Monte Carlo and $n_s = 10^8$ samples as $p_{F,MC}= 3.02 \times 10^{-3}$.
The FORM estimate of the corresponding series system probability is $p_{F_1} = 2.82 \times 10^{-3}$.
Table~\ref{t:resultsex2} demonstrates excellent agreement between the pick-freeze Monte Carlo estimates of the first-order and total-effect indices and the corresponding FORM approximations.

\begin{figure}[h]
  \centering
    \includegraphics[width=0.6\textwidth]{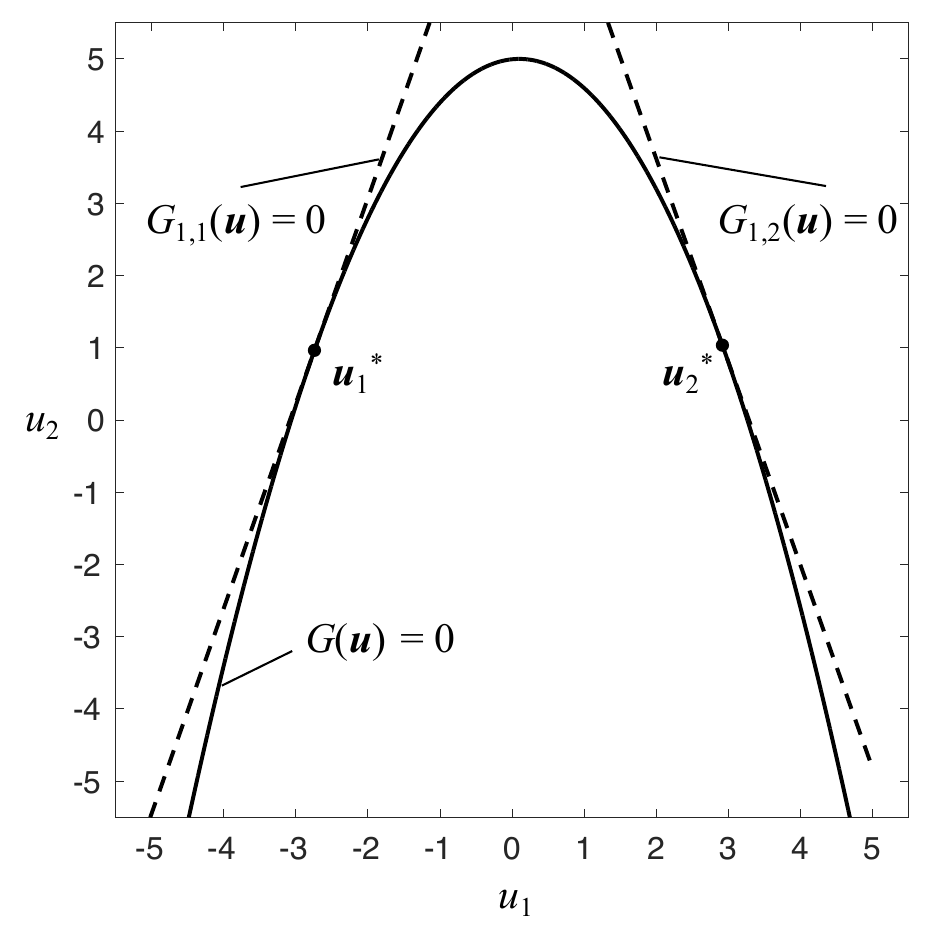}
    \caption{The limit state surface of the parabolic problem with its two design points and the corresponding linearized limit states.}
\label{f:parabolicLSF}
\end{figure}

\begin{table}
\caption{Estimates of the first-order and total-effect indices for the parabolic limit state.}
\footnotesize
\begin{center}
\begin{tabular}{llll}\hline
\noalign{\smallskip}
Method& Index &$U_1$ & $U_2$ \\
\noalign{\smallskip}\hline
\noalign{\smallskip}
MC & $S_{F,i}$ &$0.437$ & $7.77 \times 10^{-3}$  \\
& $S^T_{F,i}$ &$0.992$ & $0.564$ \\
FORM & $S_{F,i}$ &$0.430$ & $6.39 \times 10^{-3}$  \\
& $S^T_{F,i}$ &$0.994$ & $0.571$   \\
\noalign{\smallskip}\hline
\end{tabular}
\end{center}
\label{t:resultsex2}
\end{table}

\subsection{Cantilever beam-bar system}

This example, taken from \cite{song2003bounds}, considers a cantilever beam-bar system as illustrated in Fig.~\ref{f:beam-bar}.
The cantilever beam of length $2 L$ has plastic moment $M$ and is propped by an ideally rigid brittle bar of strength $T$.
A load $P$ is applied at the midpoint of the beam.
The system failure event can be described by the following LSFs:
\begin{eqnarray}  \nonumber
g_{1}(M,T,P) &=& T-5P/16 \, , \\ \nonumber
g_{2}(M,T,P) &=& M-LP \, , \\ \nonumber
g_{3}(M,T,P) &=& M-3L P/8 \, , \\ \nonumber
g_{4}(M,T,P) &=& M-L P/3 \, , \\ \nonumber
g_{5}(M,T,P) &=& M+2L T -L P \, . 
\end{eqnarray}
The minimum cut-sets of the system are $\boldsymbol{c}_1=\{1,2 \}$, $\boldsymbol{c}_2=\{3,4 \}$ and $\boldsymbol{c}_3=\{3,5 \}$; the system failure event is given by Eq.~\eqref{eq:Fgensys} with $i$-the component failure event represented by LSF $g_i$.
The length is deterministic with $L=5$ and the parameters $[M, T, P]$ are described by random variables with moments as given in Table~\ref{t:modelex4}.

We consider two cases. 
First, we model the uncertain parameters with Gaussian random variables.
The system probabilty of failure is $p_{F,\mathrm{sys}}= 7.76 \times 10^{-3}$.
As the LSFs are linear, we expect the reliability sensitivities as evaluated with Eq.~\eqref{eq:Sobol1ser} to be exact.
This is verified in Table~\ref{t:resultsex41}, where the FORM results are compared with the pick-freeze Monte Carlo estimates obtained with $n_s = 10^8$ samples, where the differences between the estimates can be attributed to sampling fluctuations.
Next, we consider the case where the strength parameters $M$ and $T$ follow the lognormal distribution and $P$ follows the normal distribution.
In this case, the probability of failure with Monte Carlo is estimated as $p_{F,MC}= 2.22 \times 10^{-4}$ and the FORM approximation is $2.61 \times 10^{-4}$.
FORM is not exact in this case as the distribution of the random  variables is non-Gaussian.
We note that the FORM approximation is evaluated through linearizing each LSF at the corresponding design point, obtained through solution of Eq.~\eqref{eq:FORMopt}.
This is expected to result in a suboptimal approximation; optimally, the joint design point of each parallel system entering the inclusion-exclusion rule of Eq.~\eqref{eq:PFsysFORM} should be identified.
The sensitivity estimates obtained by the two methods are given in Table~\ref{t:resultsex42}.
The estimates show very good agreement, although FORM slightly overestimates the first-order contributions of all three variables and underestimates the total-effect contribution of the plastic moment $M$.
It is also interesting to point out that changing the distribution of the strength parameters has a strong effect on the sensitivities; the first-order indices of $T$ and $P$ increase considerably and $P$ replaces $M$ as the variable with the most dominant total-effect contribution.

\begin{figure}[h]
  \centering
    \includegraphics[width=0.45\textwidth]{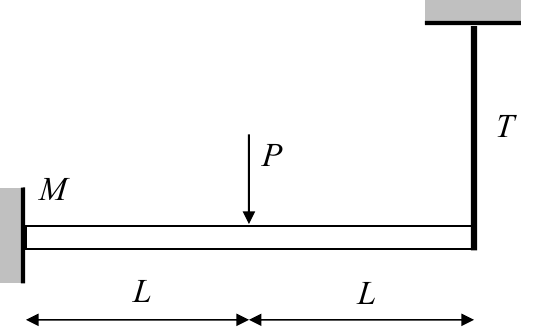}
    \caption{Cantilever beam-bar system according to \cite{song2003bounds}.}
\label{f:beam-bar}
\end{figure}

\begin{table}
\caption{Uncertain parameters of the cantilever beam-bar example.}
\footnotesize
\begin{center}
\begin{tabular}{lll}\hline
\noalign{\smallskip}
Parameter&Mean&St. Dev.\\
\noalign{\smallskip}\hline
\noalign{\smallskip}
$M$ &$1000$&$300$\\
$T$ &$110$&$20$\\
$P$ &$150$&$30$\\
\noalign{\smallskip}\hline
\end{tabular}
\end{center}
\label{t:modelex4}
\end{table}

\begin{table}
\caption{Estimates of the first-order and total-effect indices for the general system example with Gaussian inputs.}
\footnotesize
\begin{center}
\begin{tabular}{lllll}\hline
\noalign{\smallskip}
Method& Index &$M$ & $T$ & $P$\\
\noalign{\smallskip}\hline
\noalign{\smallskip}
MC & $S_{F,i}$ &$0.654 $ & $0.0103$ & $3.85 \times 10^{-3}$\\
& $S^T_{F,i}$ &$0.939$ & $0.117$ & $0.323$ \\
FORM & $S_{F,i}$ &$0.654$ & $0.0102$ & $4.05 \times 10^{-3}$ \\
& $S^T_{F,i}$ &$0.939$ & $0.117$ & $0.323$ \\
\noalign{\smallskip}\hline
\end{tabular}
\end{center}
\label{t:resultsex41}
\end{table}

\section{Concluding remarks} \label{s:conclusion}

This paper derives expressions for the first-order and total-effect reliability sensitivity indices of the FORM approximation for general system reliability problems.
These expressions can be used to estimate reliability sensitivities of general system problems with independent inputs.
They can also be used to approximate the sensitivities of nonlinear component reliability problems with multiple design points.
Numerical examples demonstrate that the proposed expressions result in highly accurate approximations of the reliability sensitivities.
Although the proposed expressions assume independent inputs, they can be employed for reliability sensitivity analysis of dependent inputs through application within the approach discussed in \cite{ehre2024variance} (cf. Remark \ref{rem:dependent}).

\begin{table}
\caption{Estimates of the first-order and total-effect indices for the general system example with non-Gaussian inputs.}
\footnotesize
\begin{center}
\begin{tabular}{lllll}\hline
\noalign{\smallskip}
Method& Index &$M$ & $T$ & $P$\\
\noalign{\smallskip}\hline
\noalign{\smallskip}
MC & $S_{F,i}$ &$0.0217 $ & $0.0238$ & $0.0114$\\
& $S^T_{F,i}$ &$0.426$ & $0.841$ & $0.947$ \\
FORM & $S_{F,i}$ &$0.0232$ & $0.0277$ & $0.0161 $ \\
& $S^T_{F,i}$ &$0.364$ & $0.841$ & $0.933$ \\
\noalign{\smallskip}\hline
\end{tabular}
\end{center}
\label{t:resultsex42}
\end{table}

The proposed expressions require the evaluation of the $2 m$-dimensional normal integral, with $m$ denoting the number of components of the system.
We have observed that available algorithms for estimating this integral tend to require a large number of samples when $m$ is large and the correlation between the components is high.
In such cases, it would be beneficial to derive approximations of the FORM sensitivities that can be computed efficiently, e.g., through application of system reliability bounds.
Another possible future research direction is to extend the work of \cite{papaioannou2022reliability} to obtain FORM approximations of the decision-theoretic sensitivities discussed in \cite{straub2022decision} for system problems.
The latter sensitivities take into account the effects of the input variables on the optimality of a decision taken based on the model output; they have the advantage of better interpretability.

\bibliographystyle{model1-num-names}
\bibliography{mybib}







\end{document}